\documentclass[11pt]{article}

\usepackage[margin=1in]{geometry}
\usepackage{amsthm,amsmath,amsfonts,amssymb}
\usepackage{mathtools}
\mathtoolsset{centercolon}

\usepackage{enumitem}

\usepackage{thm-restate}

\usepackage[pagebackref]{hyperref} 
\hypersetup{
    unicode=true, 
    pdftitle={Nearly optimal separations between communication (or query) complexity and partitions}, 
    pdfauthor={Robin Kothari}, 
    pdfnewwindow=true, 
    colorlinks=true, 
    linkcolor=blue, 
    citecolor=blue, 
    filecolor=blue, 
    urlcolor=blue 
}

\renewcommand{\backref}[1]{}

\renewcommand{\backrefalt}[4]{%
\ifcase #1 %
\or 
[p.\ #2]%
\else 
[pp.\ #2]%
\fi}

\usepackage{booktabs}

\makeatletter
\renewcommand{\paragraph}{%
  \@startsection{paragraph}{4}%
  {\z@}{2.25ex \@plus .5ex \@minus .3ex}{-1em}%
  {\normalfont\normalsize\bfseries}%
}
\makeatother

\newtheorem{theorem}{Theorem}
\newtheorem{lemma}[theorem]{Lemma}
\newtheorem{definition}[theorem]{Definition}

\theoremstyle{definition}

\newcommand{\eq}[1]{\hyperref[eq:#1]{(\ref*{eq:#1})}}
\renewcommand{\sec}[1]{\hyperref[sec:#1]{Section~\ref*{sec:#1}}}
\newcommand{\thm}[1]{\hyperref[thm:#1]{Theorem~\ref*{thm:#1}}}
\newcommand{\lem}[1]{\hyperref[lem:#1]{Lemma~\ref*{lem:#1}}}
\newcommand{\prop}[1]{\hyperref[prop:#1]{Proposition~\ref*{prop:#1}}}
\newcommand{\cor}[1]{\hyperref[cor:#1]{Corollary~\ref*{cor:#1}}}
\newcommand{\defn}[1]{\hyperref[def:#1]{Definition~\ref*{def:#1}}}
\newcommand{\fig}[1]{\hyperref[fig:#1]{Figure~\ref*{fig:#1}}}
\newcommand{\tab}[1]{\hyperref[tab:#1]{Table~\ref*{tab:#1}}}
\newcommand{\alg}[1]{\hyperref[alg:#1]{Algorithm~\ref*{alg:#1}}}
\newcommand{\app}[1]{\hyperref[app:#1]{Appendix~\ref*{app:#1}}}

\newcommand{\comment}[1]{}

\newcommand{\B}{\{0,1\}}
\newcommand{\CS}{\mathrm{CS}}

\newcommand{\UC}{\mathrm{UC}}
\newcommand{\Dcc}{D^\mathrm{cc}}

\newcommand{\X}{\mathcal{X}}
\newcommand{\Y}{\mathcal{Y}}

\newcommand{\tO}{\widetilde{O}}
\newcommand{\tOmega}{\widetilde{\Omega}}

\newcommand{\AND}{\textsc{And}}
\newcommand{\OR}{\textsc{Or}}
\newcommand{\BKK}{\textsc{Bkk}}
\newcommand{\ksum}{\textsc{k-sum}}
\newcommand{\bksum}{\textsc{bk-sum}}

\DeclareMathOperator{\poly}{poly}
\renewcommand{\th}[1]{${#1}^{\textrm{th}}$}

\renewcommand{\(}{\left(}
\renewcommand{\)}{\right)}

\newcommand{\be}{\begin{equation}}
\newcommand{\ee}{\end{equation}}
\def\ba#1\ea{\begin{align}#1\end{align}}

\begin{document}

\title{Nearly optimal separations between \\ communication (or query) complexity and partitions}

\author{
Robin Kothari\thanks{Center for Theoretical Physics, Massachusetts Institute of Technology. \texttt{rkothari@mit.edu}}
}

\date{}
\maketitle
\vspace{-1.05em}
\begin{abstract}

We show a nearly quadratic separation between deterministic communication complexity and the logarithm of the partition number, which is essentially optimal. This improves upon a recent power $1.5$ separation of G\"o\"os, Pitassi, and Watson (FOCS 2015). In query complexity, we establish a nearly quadratic separation between deterministic (and even {randomized}) query complexity and subcube partition complexity, which is also essentially optimal. 
We also establish a nearly power $1.5$ separation between quantum query complexity and subcube partition complexity, the first superlinear separation between the two measures.
Lastly, we show a quadratic separation between quantum query complexity and one-sided subcube partition complexity.

Our query complexity separations use the recent cheat sheet framework of Aaronson, Ben-David, and the author. Our query functions are built up in stages by alternating function composition with the cheat sheet construction. The communication complexity separation follows from ``lifting'' the query separation to communication complexity.
\end{abstract}

\section{Introduction}
\label{sec:intro}

\paragraph{Deterministic communication complexity.} In the standard model of communication complexity, we wish to compute a function $F:\X\times \Y \to \B$, where the  inputs $x\in\X$ and $y\in\Y$ are given to two different players, while minimizing the communication between the players. We use $\Dcc(F)$ to denote the deterministic communication complexity of $F$, the number of bits communicated in the worst case by the best deterministic protocol for the function $F$.

The partition number of $F$, denoted $\chi(F)$, is the least number of monochromatic rectangles in a partition or disjoint cover of $\X \times \Y$ (where a monochromatic rectangle is a set $A\times B$, with $A \subseteq \X$ and $B \subseteq \Y$, such that $F$ takes the same value on all elements of $A \times B$). Yao \cite{Yao79} observed that any $C$-bit communication protocol for $F$ partitions the set of all inputs $\X \times \Y$ into at most $2^C$  monochromatic rectangles, which gives us $\log\chi(F)\leq \Dcc(F)$. This turns out to be a powerful lower bound, and in fact almost all lower bound techniques for deterministic communication complexity, including the partition bound, discrepancy, fooling sets, (nonnegative) rank, and various norm-based methods \cite{JK10,JLV14,LS07}, actually lower bound $\log\chi(F)$.

In addition to being a fruitful lower bound technique, $\log\chi(F)$ also yields an upper bound on $\Dcc(F)$.  Aho, Ullman, and Yannakakis \cite{AUY83} showed that for all $F:\X\times \Y \to \B$, we have
\be 
\label{eq:AUY}
\Dcc(F) = O(\log^2\chi(F)).
\ee 
It has been a long-standing open problem to determine whether this upper bound can be improved (see, e.g., \cite[Open Problem 2.10]{KN06}). We show that the upper bound in \eq{AUY} is essentially optimal.

\begin{restatable}{theorem}{Dccvschi}
\label{thm:Dccvschi}
There exists a function $F:\X\times \Y \to \B$ with $\Dcc(F)\geq \bigl(\log\chi(F)\bigr)^{2-o(1)}$.
\end{restatable}

Until recently, the best known separation between the two measures was only by a factor of $2$ \cite{KLO99}.
Recently, G\"o\"os, Pitassi, and Watson \cite{GPW15} showed that there exists a function $F$ with $\Dcc(F) = \tOmega(\log^{1.5}\chi(F))$, where the notation $\tOmega(m)$ hides $\poly(\log m)$ factors. \tab{Dccvschi} summarizes known separations between $\Dcc$ and $\log\chi$.

\setlength{\tabcolsep}{5pt}
\renewcommand{\arraystretch}{1.3}

\begin{table}
\centering
\begin{minipage}[t]{0.6\textwidth}
 \centering
 \begin{tabular}{l@{\hskip 5em}l}\toprule
 Separation & Reference \\ \midrule
 $\Dcc(F) \geq 2 \log\chi(F)$ & \cite{KLO99} \\ 
 $\Dcc(F) = \tOmega\(\log^{1.5} \chi(F)\)$ & \cite{GPW15} \\ 
 $\Dcc(F) \geq \bigl(\log \chi(F)\bigr)^{2-o(1)}$ & \thm{Dccvschi}\\ \bottomrule
\multicolumn{2}{c}{$\Dcc(F) = O\bigl(\log \chi(F)^2\bigr)$ for all $F:\X \times \Y \to \B$}\\ \bottomrule
\end{tabular}
\caption{Known separations between deterministic communication complexity, $\Dcc(F)$, and partition number, $\chi(F)$.}
\label{tab:Dccvschi}
\end{minipage}
 \end{table}

\paragraph{Deterministic query complexity.} In the model of query complexity, we wish to compute a function $f:\B^n \to \B$ on an input $x\in \B^n$ given query access to the bits of the input, i.e., we can only access the input via a black box that accepts an index $i\in[n]$ (where $[n]:=\{1,2,\ldots,n\}$) and responds with $x_i\in\B$. The goal is to compute $f(x)$ while minimizing the number of queries made to the black box. Let $D(f)$ to denote the deterministic query complexity of $f$, the number of queries made by the best deterministic algorithm that computes $f$ correctly on all inputs.

As in communication complexity, most lower bounds for deterministic query complexity are based on the simple observation that any $d$-query algorithm computing $f$ partitions the domain $\B^n$ into at most $2^d$ monochromatic subcubes where each subcube fixes at most $d$ variables. A subcube is a restriction of the hypercube where some variables have been fixed, and it is monochromatic if $f$ takes the same value on all inputs in the subcube. This motivates defining the subcube partition complexity of $f$ as a smallest $d$ such that the domain $\B^n$ can be partitioned into at most $2^d$ monochromatic subcubes that each fix at most $d$ variables. Subcube partition complexity can also be viewed as an unambiguous version of certificate complexity as explained in \sec{prelim}, and hence we denote this measure $\UC(f)$.

Due to the observation above, we have $\UC(f) \leq D(f)$. It turns out that this lower bound is also relatively tight: for all $f:\B^n\to \B$ we have
\be 
\label{eq:DvsUC}
D(f) = O(\UC(f)^2).
\ee 

We show that this upper bound is essentially optimal.

\begin{restatable}{theorem}{DvsUC}
\label{thm:DvsUC}
There exists a total function $f$ with $D(f)\geq \UC(f)^{2-o(1)}$.
\end{restatable}

The first separation between these two measures was a power $1.261$ separation by Savick\'{y}, which was recently improved by G\"o\"os, Pitassi, and Watson \cite{GPW15} to power $1.5$. \tab{DvsUC} summarizes known separations between these measures.

\setlength{\tabcolsep}{2pt}
\renewcommand{\arraystretch}{1.3}

 \begin{table}
\centering
\begin{minipage}[t]{0.47\textwidth}
\centering
 \begin{tabular}{l@{\hskip 5em}l}\toprule
 Separation & Reference \\ \midrule
 $D(f) = \Omega (\UC(f)^{1.261})$ & \cite{Sav02} \\ 
 $D(f) = \tOmega (\UC(f)^{1.5})$ & \cite{GPW15} \\ 
 $D(f) \geq \UC(f)^{2-o(1)}$ & \thm{DvsUC}\\ \bottomrule
\multicolumn{2}{l}{$D(f) = O(\UC(f)^2)$ for all $f:\B^n \to \B$} \\ \bottomrule
\end{tabular}
\caption{Known separations between deterministic query complexity, $D(f)$, and subcube partition complexity, $\UC(f)$.}
\label{tab:DvsUC}
\end{minipage}
\hspace{1.8em}
\begin{minipage}[t]{0.47\textwidth}
\centering
 \begin{tabular}{l@{\hskip 5em}l}\toprule
 Separation & Reference \\ \midrule
 $R(f) = \Omega (\UC(f)^{1.058})$ & \cite{KRS15} \\ 
 $R(f) = \tOmega (\UC(f)^{1.5})$ & \cite{GJPW15} \\ 
 $R(f) \geq \UC(f)^{2-o(1)}$ & \thm{RvsUC}\\ \bottomrule
\multicolumn{2}{l}{$R(f) = O(\UC(f)^2)$ for all $f:\B^n \to \B$} \\ \bottomrule
\end{tabular}
\caption{Known separations between randomized query complexity, $R(f)$, and subcube partition complexity, $\UC(f)$.}
\label{tab:RvsUC}
\end{minipage}
 \end{table}

\paragraph{Randomized query complexity.} 
We can extend the query model to allow randomized algorithms in the natural way. We define the bounded-error randomized query complexity of a function $f$, $R(f)$, to be the minimum number of queries needed in the worst case by a randomized algorithm that outputs $f(x)$ on input $x$ with probability at least $2/3$.

As before, almost all lower bound techniques for randomized query complexity are upper bounded by $\UC(f)$, as shown in \cite{KRS15}. This includes the partition bounds \cite{JK10,JLV14}, approximate polynomial degree \cite{NS95}, approximate nonnegative junta degree (also known as nonnegative literal degree or conical junta degree) \cite{KLdW15}, block sensitivity \cite{Nis91}, randomized certificate complexity or fractional block sensitivity \cite{Aar06,GSS13,Tal13}, and the classical analogue of the quantum adversary bound \cite{LM08,SS06, Aar08}.

Since we obviously have $R(f) \leq D(f)$, using \eq{DvsUC} we know that $R(f) = O(\UC(f)^2)$. We show that this upper bound is also essentially optimal.

\begin{restatable}{theorem}{RvsUC}
\label{thm:RvsUC}
There exists a total function $f$ with $R(f)\geq \UC(f)^{2-o(1)}$.
\end{restatable}

The first asymptotic separation between these measures was a power $1.058$ separation by Racicot-Desloges, Santha, and the author \cite{KRS15}, which was later improved by  G\"o\"os, Jayram, Pitassi, and Watson  \cite{GJPW15} to a power $1.5$ separation. \tab{RvsUC} summarizes the known separations between these measures.

\paragraph{Quantum query complexity.} 
The query model can also be naturally extended to quantum algorithms. We denote by $Q(f)$ the bounded-error quantum query complexity of $f$, the minimum number of queries made in the worst case by a quantum algorithm that outputs $f(x)$ on input $x$ with probability at least $2/3$. (See \cite{BdW02} for a formal definition.)

As before, since $Q(f) \leq D(f)$, using \eq{DvsUC} we know that $Q(f) = O(\UC(f)^2)$. However, prior to our work no function was known for which $Q(f) \gg \UC(f)$ was known. Furthermore, the functions previously used to show separations between $D(f)$ or $R(f)$ and $\UC(f)$ do not separate $Q(f)$ from $\UC(f)$. Indeed, even the functions constructed to prove \thm{DvsUC} and \thm{RvsUC} do not separate $Q(f)$ from $\UC(f)$. Despite this, we give the first superlinear separation between $Q(f)$ and $\UC(f)$.

\begin{restatable}{theorem}{QvsUC}
\label{thm:QvsUC}
There exists a total function $f$ with $Q(f) \geq \UC(f)^{1.5-o(1)}$.
\end{restatable}

We are also able to show an improved separation between quantum query complexity and one-sided subcube partition complexity, denoted by $\UC_1(f)$, which is similar to subcube partition complexity except that we only need to partition the $1$-inputs using monochromatic subcubes. 

For this measure, the quadratic upper bound $D(f) = O(\UC_1(f)^2)$ still holds \cite[Proposition 5]{Goo15}, and hence $Q(f) = O(\UC_1(f)^2)$. We show this upper bound is optimal up to log factors, qualitatively improving upon \cite{GPW15} and \cite{GJPW15} who proved the same result for deterministic and randomized query complexity respectively.

\begin{restatable}{theorem}{QvsUCone}
\label{thm:QvsUC1}
There exists a total function $f$ with $Q(f) = \tOmega(\UC_1(f)^{2})$.
\end{restatable}

\section{High-level overview}

We now provide a high-level overview of the separations shown.

\paragraph{Deterministic communication complexity.} We prove \thm{Dccvschi} by showing the analogous separation in query complexity (\thm{DvsUC}) and ``lifting'' the result to communication complexity, which is also the  strategy used in \cite{GPW15}. Essentially, the deterministic simulation theorem of \cite{GPW15} provides a black-box way of converting a query separation between $D(f)$ and $\UC(f)$ to a separation between $\Dcc(F)$ and $\log\chi(F)$. The theorem weakens the separation by log of the input size of $f$, but with a suitable choice of parameters this is negligible compared to the $o(1)$ term in the separation.

\paragraph{Deterministic query complexity.}
To prove \thm{DvsUC}, we use the recently introduced cheat sheet framework \cite{ABK15} and the commonly used technique of function composition. Before describing the construction, we need to define some notation. For any functions $f:\B^n\to\B$ and $g:\B^m\to \B$, we define the composed function $f \circ g$ to be the function on $mn$ bits whose output on $y = (y_{11},\ldots, y_{1m}, \ldots, y_{n1},\ldots, y_{nm})$ is $f(g(y_{11},\ldots, y_{1m}),\ldots ,g(y_{n1}, \ldots , y_{nm}))$. Let $\AND_n$ and $\OR_n$ denote the $\AND$ and $\OR$ function on $n$ bits respectively. 
For any function $f$, we use $f_\CS$ to denote the ``cheat sheet version'' of $f$, a new total Boolean function constructed from $f$. (We review the cheat sheet framework in \sec{cheat}.)

The usual application of function composition (see \cite{Tal13} for examples) is to construct a constant-sized function $f$ with desirable properties and then compose $f$ with itself several times. Our construction does not use this strategy. Instead,  we construct our function in stages starting with the function $f_0 = \AND_n$ that achieves no separation between $D(f)$ and $\UC(f)$. We then compose the function with $\OR_n$, construct the cheat sheet version, and then compose with $\AND_n$, to obtain the function $f_1 = \AND_n\circ (\OR_n \circ \AND_n)_\CS$, which achieves a power $3/2$ separation between $D(f)$ and $\UC(f)$. Repeating this construction once more yields $f_2 = \AND_n\circ (\OR_n \circ \AND_n\circ (\OR_n \circ \AND_n)_\CS)_\CS$, which achieves a power $5/3$ separation, and so on. The function $f_k$ achieves a $(2k+1)/(k+1)$ separation, which yields a $2-o(1)$ separation if we choose $k$ to be a slow growing function of $n$.

\paragraph{Randomized query complexity.}
The function constructed above also yields the separation in \thm{RvsUC} with slightly worse parameters. The analysis of the constructed function is similar since deterministic and randomized query complexities behave similarly with respect to the cheat sheet technique and with respect to composition with the $\AND$ and $\OR$ functions.

\paragraph{Quantum query complexity.}
Lastly, we establish the quantum separations  using two functions introduced by Aaronson, Ben-David and the author \cite{ABK15}: the \textsc{Block-$k$-sum-of-$k$-sums} function, which we denote $\BKK$, and the \textsc{Block-$k$-sum} function, which we denote $\bksum$. The function $\BKK_\CS$ yields the separation in \thm{QvsUC1}. The separation in \thm{QvsUC} requires a function constructed in stages again. The first function is $f_1 = \AND \circ \BKK_\CS$, which achieves a power $5/4$ separation, the next is $f_2=\AND_n \circ (\bksum_n \circ f_1)$, which achieves a power $4/3$ separation and so on. The function $f_k$ achieves a power $(3k+2)/(2k+2)$ separation.

\section{Preliminaries}
\label{sec:prelim}

\paragraph{Communication complexity.} The only communication complexity measures we need are $\Dcc(F)$ and $\chi(F)$, which were defined in \sec{intro}. The interested reader is referred to \cite{KN06,Juk12} for more formal definitions of these measures.

\paragraph{Query complexity.} For more formal definitions of measures introduced in \sec{intro}, the reader is referred to the survey by Buhrman and de Wolf \cite{BdW02}. The only measure not covered in the survey is subcube partition complexity, which is explained in detail in \cite{KRS15}.

Subcube partition complexity can also be viewed as unambiguous certificate complexity and we use this perspective in this paper. To explain this, let us begin with certificate complexity. 

A certificate for an input $x\in \B^n$ is a subset $S\subseteq [n]$ of indices and claimed values for these bits, such that $x$ is consistent with the certificate and any input $y$ consistent with the certificate satisfies $f(x)=f(y)$. In other words, a certificate for $x$ is a partial assignment of bits consistent with $x$ such that any other string consistent with this partial assignment has the same function value as $x$. 
For $b\in\B$, the $b$-certificate complexity of $f$, denoted $C_b(f)$, is the size of the smallest certificate for $x$ maximized over all inputs with $f(x)=b$. The certificate complexity of $f$, $C(f)$, is defined as $C(f) := \max\{C_0(f),C_1(f)\}$. Alternately, $C_1(f)$ is the smallest $w$ such that $f$ can be written as a width-$w$ DNF, i.e., a DNF in which each term contains at most $w$ variables. Similarly, $C_0(f)$ corresponds to CNF width.

Unambiguous certificate complexity is defined similarly, except we require the set of certificates to be unambiguous, i.e., at most one certificate from the set of all certificates should work for a given input. In other words, the unambiguous $1$-certificate complexity of $f$ is the minimum $w$ such that $f$ can be written as a width-$w$ DNF in which at most one term evaluates to $1$ on any input. Similar to certificate complexity, we denote unambiguous $b$-certificate complexity by $\UC_b(f)$ and define $\UC(f) := \max\{\UC_0(f),\UC_1(f)\}$. Clearly, since unambiguous certificates are more restricted than certificates, we have for $b \in \B$,  $C_b(f) \leq \UC_b(f)$ and $C(f) \leq \UC(f)$.

For example, consider the $\OR_n$ function on $n$ bits defined as $\bigvee_{i\in[n]} x_i$. Clearly $C_0(\OR_n)=n$ since we must examine all $n$ bits to be sure that all $x_i=0$. On the other hand, $C_1(\OR_n)=1$ since the location of any $1$ in the input is a certificate. 
Obviously $\UC_0(\OR_n)$ remains $n$. However, a single $1$ in the input is not an unambiguous $1$-certificate since inputs with multiple $1$s would have multiple valid certificates. In other words, although $\bigvee_{i\in[n]} x_i$ is a valid DNF representation of $\OR_n$, it is not  unambiguous since several terms can simultaneously be $1$. So consider the following  DNF:
\begin{equation}
\label{eq:ordnf}
\OR_n(x)=x_1 \vee \overline{x_1}x_2 \vee \overline{x_1}\overline{x_2} x_3 \vee \cdots \vee \overline{x_1x_2\cdots x_{n-1}}x_n
\end{equation}
This DNF is unambiguous since any term evaluating to $1$ prevents other terms from evaluating to $1$.
Thus we have $\UC_1(\OR_n)\leq n$. Although this result in trivial because $\UC_1(f)\leq n$ for any $n$-bit function $f$, this DNF representation of $\OR_n$ will be useful to us later later because it has the property that every unambiguous certificate has only one unnegated index $x_i$.

\paragraph{Composition theorems.} Composition theorems relate the complexity of composed functions with the complexities of the individual functions. For example, for all Boolean functions $f$ and $g$, $D(f\circ g)=D(f)D(g)$ \cite{Tal13,Mon14}. In our construction we will repeatedly compose functions with $\AND_n$ and $\OR_n$, and hence we need to understand the complexities of the resulting functions.

\begin{lemma}[AND/OR composition]
\label{lem:composition}
For any total Boolean function $f$, the following bounds hold:

\vspace{0.5em}
\centering
\begin{minipage}[t]{0.5\textwidth}
\begin{itemize}[noitemsep]
\item $D(\AND_n\circ f) = n D(f)$
\item $R(\AND_n\circ f)=\Omega(n R(f))$
\item $Q(\AND_n\circ f)=\Omega(\sqrt{n} Q(f))$
\item $C_0(\AND_n \circ f) \leq C_0(f)$
\item $C_1(\AND_n \circ f) \leq n C_1(f)$
\item $\UC_0(\AND_n \circ f) \leq \UC_0(f)+(n-1)\UC_1(f)$
\item $\UC_1(\AND_n \circ f) \leq n \UC_1(f)$
\end{itemize}
\end{minipage}%
\begin{minipage}[t]{0.5\textwidth}
\begin{itemize}[noitemsep]
\item $D(\OR_n\circ f) = n D(f)$
\item $R(\OR_n\circ f)=\Omega(n R(f))$
\item $Q(\OR_n\circ f)=\Omega(\sqrt{n} Q(f))$
\item $C_0(\OR_n\circ f) \leq nC_0(f)$
\item $C_1(\OR_n \circ f) \leq C_1(f)$
\item $\UC_0(\OR_n \circ f) \leq n\UC_0(f)$
\item $\UC_1(\OR_n \circ f) \leq (n-1)\UC_0(f)+\UC_1(f)$
\end{itemize}
\end{minipage}
\end{lemma}
\begin{proof}
We prove the claims for the function $\AND_n\circ f$. Similar reasoning proves the analogous claims for the function $\OR_n \circ f$.

The first property follows from the fact that $D(f\circ g) = D(f) D(g)$ for any Boolean functions $f$ and $g$~\cite{Tal13,Mon14}. $R(\AND_n\circ f)=\Omega(n R(f))$ was recently proved by \cite{GJPW15}. $Q(\AND_n\circ f)=\Omega(\sqrt{n} Q(f))$ because $Q(f\circ g) = \Theta(Q(f) Q(g))$ for any Boolean functions $f$ and $g$ \cite{HLS07,Rei11,LMR+11} and we know that $Q(\AND_n)=Q(\OR_n)=\Theta(\sqrt{n})$ \cite{Gro96,BBBV97}.

We have $C_0(\AND_n \circ f) \leq C_0(f)$ since a $0$-certificate for $\AND_n$ is a $0$-input to it, which corresponds to an instance of $f$ that evaluates to $0$. On the other hand, by certifying that all $n$ instances of $f$ evaluate to $1$, we can certify $\AND_n \circ f$ evaluates to $1$, and hence $C_1(\AND_n \circ f) \leq n C_1(f)$.

We can unambiguously certify that $\AND_n\circ f$ evaluates to $0$ by unambiguously certifying the value of the first (from the left)  $1$-input to the $\AND_n$ gate and unambiguously certifying that all previous inputs are $0$.
This is the same idea used to construct the unambiguous DNF for $\OR_n$ in \eq{ordnf}. This construction gives $\UC_0(\AND_n \circ f) \leq \UC_0(f)+(n-1)\UC_1(f)$. We can unambiguously certify that $\AND_n\circ f$ evaluates to $1$ by providing unambiguous $1$-certificates for all $n$ instances of $f$. This gives $\UC_1(\AND_n \circ f) \leq n \UC_1(f)$.
\end{proof}

\section{Cheat sheet framework}
\label{sec:cheat}

We now overview the recently introduced cheat sheet framework \cite{ABK15}. The framework as presented in \cite{ABK15} is more general and can fulfill different objectives such as making partial functions total. We present a restricted version of the framework that only works for total functions. We use the framework because it makes $1$-certificates unambiguous in a natural way.

\begin{definition}[Cheat sheet version of a total function]
\label{def:cheat}
Let $f:\B^N\to\B$ be a function, $c = 10 \log N$ and $m=10 C(f)\log^2 N$. 
Then the \emph{cheat sheet version of $f$}, denoted $f_\CS$, is a total function
\begin{equation}
f_\CS:(\B^N)^c\times(\B^{m})^{2^c}\to\B.
\end{equation}
Let the input be written as $(x^1,x^2,\ldots, x^c,Y_1,Y_2,\ldots, Y_{2^c})$, where for all $i\in [N]$, $x^i \in \B^N$ and for all $j\in[2^c]$, $Y_j \in \B^{m}$. 
Let $\ell_i = f(x^i)$ and $\ell\in[2^c]$ be the positive integer corresponding to the binary string $\ell_1\ell_2\ldots\ell_c$. Then we define the value of $f_\CS(x^1,x^2,\ldots, x^c,Y_1,Y_2,\ldots, Y_{2^c})$ to be $1$ if and only if $Y_\ell$ contains certificates for $f(x^i) = \ell_i$ for all $i\in[c]$.
\end{definition}

Informally, the cheat sheet construction takes any total function $f$ and converts it into a new total function $f_\CS$ in the following way. An input to the new function $f_\CS$ first contains $c=10\log N$ inputs to $f$ and then a vast array of size $2^c$ of cells of size $m$ bits. The outputs of these $c$ inputs to $f$ is a bit string $\ell_1\ell_2\ldots\ell_c$  of length $c$ that represents an integer $\ell \in [2^c]$ in the natural way. We treat this integer $\ell$ as an address into this array of size $2^c$ and say that these $c$ inputs to $f$ point to the \th{\ell} cell of the array. At the \th{\ell} cell of the array we require certificates certifying that this was indeed the cell pointed to by the $c$ inputs to $f$. In other words, we require certificates certifying that $f(x^i)$, the output of $f$ acting on the \th{i} input, is indeed equal to $\ell_i$ for all $i \in [c]$. Since a certificate for a single $f$ consists of $C(f)$ pointers to the input, a certificate is of size $C(f)\log N$ bits, and hence $c$ certificates are of size $m=C(f)c\log N = 10 C(f) \log^2 N$.
The function $f_\CS$ is defined to be $1$ if and only if the input satisfies this property, i.e., if the cell pointed to by the $c$ instances does indeed contain certificates certifying it is the correct cell.

This construction preserves the complexity of $f$ with respect to some measures. For example, $D(f_\CS)$ equals $D(f)$ up to log factors. The upper bound uses the natural algorithm for $f_\CS$: the deterministic algorithm first computes the $c$ copies of $f$ on inputs $x^1$ to $x^c$ and finds the cell pointed to by these $c$ inputs. Then it checks if the certificates in this cell certify that this is the right cell. This requires $cD(f)$ queries to compute the $c$ copies, $m$ queries to read the contents of the cell and $cC(f)$ queries to check if the certificates are all correct. Overall this uses $O(cD(f))$ queries. We also have $D(f_\CS) = \Omega(D(f))$, because intuitively if an algorithm cannot compute $f$ it has no hope of finding the cheat sheet since that would require solving $c$ copies of $f$ or searching in an array of size $n^{10}$. Similarly, many measures behave as expected under cheat sheets, and we show this below.

\begin{lemma}[Complexity of cheat sheet functions]
\label{lem:cheat}
For any total function $f:\B^N\to \B$, if $f_\CS:\B^{N'}\to\B$ denotes the cheat sheet version of $f$ as defined in \defn{cheat}, then we have the following upper and lower bounds:
\begin{itemize}[noitemsep,topsep=5pt]
\item $D(f_\CS)=\Omega(D(f))$
\item $R(f_\CS)=\Omega(R(f)/\log^2 N)$
\item $Q(f_\CS)=\Omega(Q(f))$
\item $C_0(f_\CS) = O(C(f)\log^2 N)$
\item $C_1(f_\CS) = O(C(f) \log^2 N)$
\item $\UC_0(f_\CS) = O(\UC(f) \log^2 N)$
\item $\UC_1(f_\CS) = O(C(f) \log^2 N)$
\item $N' = O(N^{12})$
\end{itemize}
\end{lemma}
\begin{proof}
We have $D(f_\CS)=\Omega(D(f))$ \cite[Lemma 21]{ABK15}, 
$R(f_\CS)=\Omega(R(f)/\log^2 N)$ \cite[Lemma 6]{ABK15}, and 
$Q(f_\CS)=\Omega(Q(f))$ \cite[Lemma 12]{ABK15}.

We have $C_0(f_\CS) = O(C(f) \log^2 N)$ because a valid $0$-certificate for $f_\CS$ can first certify the $c$ outputs to $f$, which requires $O(cC(f))$ queries. This points to a cell $\ell$. The certificate can then contain the contents of cell $\ell$ of size $O(C(f)\log^2 N)$ and the locations pointed to (and the bits contained at these locations) by the certificates in cell $\ell$.  After querying this cell and all the locations pointed to by the certificates in this cell, it can be determined with no further queries if this cell is incorrectly filled. We have $C_1(f_\CS) = O(C(f) \log^2 N)$ since the location of the correct cell and the pointers within that cell along with the bits they point to forms a $1$-certificate.

We have $\UC_0(f_\CS) = O(\UC(f) \log^2 N)$ using the same argument as for certificate complexity. We first certify the $c$ outputs to $f$ unambiguously using unambiguous certificates of size $\UC(f)$. This points to a cell $\ell$. The  certificate also contains the contents of cell $\ell$ and the locations pointed to (and the bits at these locations) by the certificates in cell $\ell$. This certificate is unambiguous because this certificate evaluating to true prevents any other certificate from evaluating to true.  To see this, note that if another certificate tries to certify a different value of $\ell$ then this will be an invalid certificate. If the certificate claims the same value of $\ell$, then it must use the same certificates for the $c$ instances of $f$  because we used unambiguous certificates and hence there is only one valid certificate for each $f(x^i)=\ell_i$. Now if the other certificate has the same value of $\ell$ but different claimed values for the contents of the \th{\ell} cell or the locations pointed to by the cell, this will be inconsistent with the actual input since our original certificate was consistent with the input. 

We have $\UC_1(f_\CS) = O(C(f) \log^2 N)$. For this case an unambiguous certificate will contain only the contents of cell $\ell$ and the locations pointed to by the certificates in cell $\ell$ along with the bits contained at these locations. This is identical to the $1$-certificate we constructed above. Since this is clearly a valid certificate, we only need to show it is unambiguous, i.e., that if this certificate evaluates to true, all other certificates must fail. If another certificate has a different value of $\ell$, then its contents will not be able to certify that the output of the $c$ functions equals $\ell$ and the certificate will be rejected. On the other hand, if the other certificate has the same value of $\ell$ but different claimed values for the contents of the cell or the locations pointed to by the cell, this will be inconsistent with the input since our original certificate was consistent with the input. 

Lastly, we need to upper bound the input size of $f_\CS$. From \defn{cheat} we know the input size is $cN+m2^c = 10N\log N+ 10 N^{10} C(f)\log^2 N = O(10N^{11}\log^2 N)=O(N^{12})$.
\end{proof}

\section{Randomized query complexity vs.\ subcube partitions and deterministic communication vs.\ partition number}
\label{sec:rquery}

\paragraph{Randomized query complexity vs.\ subcube partitions.}
We now establish the following theorem which implies \thm{RvsUC}, which in turn implies \thm{DvsUC}.

\begin{theorem}
\label{thm:main}
For every $k \geq 0$, there exists a total Boolean function $f_k:\B^{N_k}\to\B$, such that 
$R(f_k)=\tOmega(n^{2k+1})$ and $\UC(f_k) = \tO(n^{k+1})$. Hence there is a function $f$ with $R(f)\geq \UC(f)^{2-o(1)}$.
\end{theorem}

\begin{proof}
Let $f_0 = \AND_n$ and $f_k$ be defined  inductively as $f_k := \AND_n \circ (\OR_n \circ f_{k-1})_\CS$, i.e., $f_k$ is the function obtain by composing $\AND_n$ with the cheat sheet version of $\OR_n$ composed with $f_{k-1}$. 

We prove the claim by induction on $k$. The induction hypothesis and the base case, $f_0 = \AND_n$, are presented below, where $N_k$ is the input size of the function $f_k$.
\begin{center}
\begin{minipage}[t]{0.5\textwidth}
\qquad \underline{Induction hypothesis ($f_k$)}
\begin{itemize}[noitemsep,topsep=2pt]
\item $N_k = O(n^{25^k})$
\item $D(f_k)=\tOmega (n^{2k+1})$
\item $R(f_k)=\tOmega (n^{2k+1})$
\item $C_0(f_k) = \tO(n^{k})$
\item $C_1(f_k) = \tO(n^{k+1})$
\item $\UC_0(f_k) = \tO(n^{k+1})$
\item $\UC_1(f_k) = \tO(n^{k+1})$
\end{itemize}
\end{minipage}%
\begin{minipage}[t]{0.3\textwidth}
\underline{Base case ($f_0 = \AND_n$)}
\begin{itemize}[noitemsep,topsep=5pt]
\item $N_0 = n$
\item $D(f_0) = n$
\item $R(f_0)=\Omega(n)$
\item $C_0(f_0) \leq 1$
\item $C_1(f_0) \leq n$
\item $\UC_0(f_0) \leq n$
\item $\UC_1(f_0) \leq n$
\end{itemize}\end{minipage}
\end{center}
The complexities of $f_0=\AND_n$ are straightforward to show and also follow from the general composition lemma (\lem{composition}) by letting $f$ be the one-bit identity function. Clearly the base case is consistent with the induction hypothesis.

We now show that the induction hypothesis for $f_k$ implies the same for $f_{k+1}$. First we upper bound the input size of $f_{k+1} = \AND_n \circ (\OR_n \circ f_{k})_\CS$. Since the input size of $f_k$ is $O(n^{25^k})$, the input size of $\OR_n \circ f_{k}$ is  $O(n^{25^k+1})$ and the input size of $(\OR_n \circ f_{k})_\CS$ is $O(n^{12(25^k+1)})$ (from \lem{cheat}). Hence the input size of $f_{k+1}$ is $O(n^{12(25^k+1)+1})=O(n^{12(25^k)+13})=O(n^{25^{k+1}})$.

The deterministic query complexity of $f_{k+1}$ can be lower bounded as follows:
\begin{align}
D(f_{k+1}) &= D( \AND_n \circ (\OR_n \circ f_{k})_\CS) = n D((\OR_n \circ f_{k})_\CS) 
= \Omega(n D(\OR_n \circ f_{k})) = \tOmega(n^{2k+3}),
\end{align} 
where we used \lem{composition} and \lem{cheat} to compute the relevant measures. The same calculation also works for $R(f_{k+1})$ up to log factors since $R(f)$ and $D(f)$ behave similarly in the aforementioned lemmas up to log factors.
Similarly using \lem{composition} and \lem{cheat} we have 
\begin{align}
C_0(f_{k+1}) &= C_0( \AND_n \circ (\OR_n \circ f_{k})_\CS) \leq  C_0((\OR_n \circ f_{k})_\CS) 
= \tO(C(\OR_n \circ f_{k})) = \tO(n^{k+1}) \   \mathrm{and}\\
C_1(f_{k+1}) &= C_1( \AND_n \circ (\OR_n \circ f_{k})_\CS) \leq n  C_1((\OR_n \circ f_{k})_\CS) 
= \tO(n C(\OR_n \circ f_{k})) = \tO(n^{k+2}).
\end{align} 

In these bounds we do not differentiate between $\log N_k$ and $\log n$ because they are asymptotically equal, since $\log N_k = 25^k \log n = O(\log n)$. Finally, using \lem{composition} and \lem{cheat} again we have 
\begin{align}
\UC_0(f_{k+1}) &= \UC_0( \AND_n \circ (\OR_n \circ f_{k})_\CS) \leq  \max \bigl\{\UC_0((\OR_n \circ f_{k})_\CS),n\UC_1((\OR_n \circ f_{k})_\CS) \bigr\}\nonumber\\
&= \tO\( \max \bigl\{\UC(\OR_n \circ f_{k}),nC(\OR_n \circ f_{k}) \bigr\} \)
= \tO(n^{k+2}) \qquad\mathrm{and}\\
\UC_1(f_{k+1}) &= \UC_1( \AND_n \circ (\OR_n \circ f_{k})_\CS) \leq n\UC_1((\OR_n \circ f_{k})_\CS)
\nonumber\\
&= \tO(n C(\OR_n \circ f_{k})) = \tO(\max\bigl\{ nC_0(f_k),C_1(f_k)\bigr\}
= \tO(n^{k+2}). 
\end{align} 
This completes the induction and establishes the first part of the theorem.

For the second part, since $R(f_k) = \tOmega(n^{2k+1})$ and $\UC(f_k) = \tO(n^{k+1})$, we have $R(f_k) = \tOmega(\UC(f_k)^{2-{\frac{1}{k+1}}})$. Since we treated $k$ as a constant, our notation hides constant and $\log n$ factors that depend only on $k$, i.e., we only get $R(f_k) \geq \(\UC(f_k)^{2-{\frac{1}{k+1}}}\)\Big/\({h_1(k)\log^{h_2(k)}n}\)$ for some functions $h_1(k)$ and $h_2(k)$. But we can always choose $k$ to be a slow growing function of $n$ so that these terms are negligible. This  yields the desired separation $R(f)\geq \UC(f)^{2-o(1)}$.
\end{proof}

Clearly \thm{DvsUC} and \thm{RvsUC} follow from this, which we restate for convenience.

\DvsUC*

\RvsUC*

\paragraph{Deterministic communication vs.\ partition number.}
We now show \thm{Dccvschi} by lifting the previous separation to communication complexity.

From \thm{RvsUC}, we have a function $f:\B^N \to \B$ such that $R(f)\geq \UC(f)^{2-o(1)}$, which implies $D(f)\geq \UC(f)^{2-o(1)}$. G\"o\"os, Pitassi, and Watson \cite{GPW15} show that for any function $f$, there is a corresponding communication problem $F$ such that 
\be 
\Dcc(F) = \Omega(D(f) \log N) = \Omega(D(f)).
\ee
On the other hand, as explained in \cite{GPW15}, we also have 
\be 
\log\chi(F) = O(\UC(f) \log N)=\tO(\UC(f)),
\ee
where we used the fact that our function has $N=n^{25^k}$, where $k$ is a slow growing function of $n$, and hence $\log N = 25^k \log n = O(\log^2 n) = O(\log^2 \UC(f))$. 

Since the conversion to communication complexity only weakens the result by log factors, the separation $D(f)\geq \UC(f)^{2-o(1)}$ immediately yields 
\be \Dcc(F)\geq \bigl(\log\chi(F)\bigr)^{2-o(1)},\ee 
 which establishes \thm{Dccvschi}:

\Dccvschi*

\section{Quantum query complexity vs.\ subcube partitions}
\label{sec:qquery}

In this section we establish \thm{QvsUC} and \thm{QvsUC1}.  To show this we require a function $\BKK$ from \cite[Theorem 10]{ABK15} with the following properties.
\begin{lemma}
\label{lem:bkk}
There exists a total function $\BKK:\B^{n^2} \to \B$ such that $C(\BKK)=\tO(n)$ and $Q(\BKK)=\tOmega(n^2)$.
\end{lemma}

We are now ready to prove \thm{QvsUC1}, restated for convenience:
\QvsUCone*

\begin{proof}
Let $f=\BKK$. Then using \lem{cheat} we know that $Q(f_\CS) = \Omega(Q(f)) = \tOmega(n^2)$ and $\UC_1(f_\CS) = O(C(f)\log^2 n) = \tO(n)$.
\end{proof}

To show \thm{QvsUC}, we need another function $\bksum$, which is a variant of the $\ksum$ problem. It has the interesting property that any certificates for it consists essentially of input bits set to $0$ and very few input bits set to $1$. As shown in the proof of \cite[Theorem 10]{ABK15}, we have the following. (More precisely, our version of $\bksum$ swaps the roles of zeros and ones compared to the function of \cite{ABK15}, but this does not affect its quantum query complexity.)

\begin{lemma}
\label{lem:ksum}
There exists a total function $\bksum:\B^{n} \to \B$ with $Q(\bksum)=\tOmega(n)$ such that for any function $f$, we have $C(\bksum\circ f) = O(nC_0(f)+C_1(f)\log^3 n)$.
\end{lemma}

In our construction, we repeatedly compose $\bksum$ with other functions and hence we need to understand the behavior of $\bksum$ under composition, analogous to \lem{composition} for $\AND$ and $\OR$.

\begin{lemma}[$\bksum$ composition]
\label{lem:ksumcomp}
For any function $f$, the following bounds hold:
\begin{itemize}[noitemsep,topsep=2pt]
\item $Q(\bksum_n\circ f)=\tOmega(n Q(f))$
\item $C(\bksum_n \circ f) = O(nC_0(f)+C_1(f)\log^3 n)$.
\item $\UC(\bksum_n \circ f) \leq n \UC(f)$
\end{itemize}
\end{lemma}

\begin{proof}
The first lower bound follows because $Q(f\circ g) = \Theta(Q(f) Q(g))$ for any Boolean functions $f$ and $g$ \cite{HLS07,Rei11,LMR+11} and we have $Q(\bksum)=\tOmega(n)$ from \lem{ksum}. The second relation follows from \lem{ksum}. Lastly, $\UC(\bksum_n \circ f) \leq n \UC(f)$ because this holds for any $n$-bit function. Any function $h\circ f$ can be unambiguously certified by showing all the outputs to $f$ and providing unambiguous certificates for each output.
\end{proof}

We are now ready to establish the following theorem, which implies \thm{QvsUC}. This proof mimics the proof structure of \thm{main} and reuses several ideas.

\begin{theorem}
For every $k \geq 0$, there exists a total Boolean function $f_k:\B^{N_k}\to\B$, such that 
$Q(f_k)=\tOmega(n^{1.5k+1})$ and $\UC(f_k) = \tO(n^{k+1})$. Hence there is a function $f$ with $Q(f)\geq \UC(f)^{1.5-o(1)}$.
\end{theorem}

\begin{proof}
Let $f_1 = \AND_n \circ \BKK_\CS$, where $\BKK$ is the function on $n^2$ bits in \lem{bkk}. Let $f_k$ be defined  inductively as $f_k := \AND_n \circ (\bksum_n \circ f_{k-1})_\CS$, i.e., $f_k$ is the function obtain by composing $\AND_n$ with the cheat sheet version of $\bksum_n$ composed with $f_{k-1}$. 

We prove the claim by induction on $k$. The induction hypothesis and the base case, $f_1 = \AND_n \circ \BKK_\CS$, are presented below, where $N_k$ is the input size of the function $f_k$.
\begin{center}
\begin{minipage}[t]{0.5\textwidth}
\qquad \underline{Induction hypothesis ($f_k$)}
\begin{itemize}[noitemsep,topsep=2pt]
\item $N_k = O(n^{25^k})$
\item $Q(f_k)=\tOmega (n^{1.5k+1})$
\item $C_0(f_k) = \tO(n^{k})$
\item $C_1(f_k) = \tO(n^{k+1})$
\item $\UC(f_k) = \tO(n^{k+1})$
\end{itemize}
\end{minipage}%
\begin{minipage}[t]{0.3\textwidth}
\underline{Base case ($f_1 = \AND_n \circ \BKK_\CS$)}
\begin{itemize}[noitemsep,topsep=4pt]
\item $N_1 = O(n^{25})$
\item $Q(f_1)=\tOmega(n^{2.5})$
\item $C_0(f_1) = \tO(n)$
\item $C_1(f_1) =\tO(n^2)$
\item $\UC(f_1) =\tO(n^2)$
\end{itemize}\end{minipage}
\end{center}

Let us first compute the complexities of $f_1 = \AND_n \circ \BKK_\CS$ and verify that they are consistent with the induction hypothesis. First note that the input size of $\BKK$ is $n^2$, and thus the input size of $\BKK_\CS$ is $O(n^{24})$ (from \lem{cheat}), and hence the input size of $f_1 = \AND_n \circ \BKK_\CS$ is $O(n^{25})$. We have $Q(f_1)=\tOmega(n^{2.5})$ since $Q(f_1)=\Omega(\sqrt{n}Q(\BKK_\CS))=\tOmega(n^{2.5})$. 
The other inequalities follow straightforwardly from \lem{composition} and \lem{bkk}.
We have $C_0(f_1) = \tO(n)$ since $C_0(\AND_n \circ \BKK_\CS) \leq C_0(\BKK_\CS)= \tO(n)$. We have $C_1(f_1) = \tO(n^2)$ since $C_1(\AND_n \circ \BKK_\CS) \leq nC_1(\BKK_\CS)= \tO(n^2)$. Lastly, we have  $\UC(f_1) = \UC(\AND_n \circ \BKK_\CS) \leq \UC_0(\BKK_\CS) + n\UC_1(\BKK_\CS) \leq D(\BKK_\CS) + \tO(nC(\BKK)) = \tO(D(\BKK)+nC(\BKK)) = \tO(n^2)$.

We now show that the induction hypothesis for $f_k$ implies the same for $f_{k+1}$. The input size calculation is identical to that in \thm{main} and hence we do not repeat it. 
The quantum  query complexity of $f_{k+1}$ can be lower bounded as follows:
\begin{align}
Q(f_{k+1}) &= Q( \AND_n \circ (\bksum_n \circ f_{k})_\CS) =\Omega(\sqrt{n} Q((\bksum_n \circ f_{k})_\CS) )
\nonumber \\
&=\Omega(\sqrt{n} Q(\bksum_n \circ f_{k}) )= \Omega(n^{1.5} Q(f_{k})) = \tOmega(n^{1.5(k+1)+1}),
\end{align} 
where we used \lem{composition}, \lem{cheat}, \lem{ksum}, and \lem{ksumcomp} to compute the relevant measures. 

Similarly we have
\begin{align}
C_0(f_{k+1}) &= C_0( \AND_n \circ (\bksum_n \circ f_{k})_\CS) \leq  C_0((\bksum_n \circ f_{k})_\CS) 
= \tO(C(\bksum_n \circ f_{k})) \nonumber \\
&= \tO(nC_0(f_k)+C_1(f_k)) = \tO(n^{k+1}) \qquad   \mathrm{and}\\
C_1(f_{k+1}) &= C_1( \AND_n \circ (\bksum_n \circ f_{k})_\CS) \leq n  C_1((\bksum_n \circ f_{k})_\CS) = \tO(n C(\bksum_n \circ f_{k}))  \nonumber \\
&= \tO(n^2C_0(f_k)+nC_1(f_k)) = \tO(n^{k+2}).
\end{align} 
Finally, using \lem{composition}, \lem{cheat}, \lem{ksum}, and \lem{ksumcomp} again we have 
\begin{align}
\UC(f_{k+1}) &= \UC( \AND_n \circ (\bksum_n \circ f_{k})_\CS) \nonumber \\
& = O\(  \max \bigl\{\UC_0((\bksum_n \circ f_{k})_\CS),n\UC_1((\bksum_n \circ f_{k})_\CS) \bigr\}\)\nonumber\\
&= \tO\( \max \bigl\{\UC(\bksum_n \circ f_{k}),nC(\bksum_n \circ f_{k}) \bigr\} \) \nonumber \\
&= \tO\( \max \bigl\{n\UC(f_{k}),n^2C_0(f_k)+nC_1(f_k) \bigr\} \)
= \tO(n^{k+2}).
\end{align} 
This completes the induction and establishes the first part of the theorem. Using a similar argument in \thm{main}, this yields a function with $Q(f)\geq \UC(f)^{1.5-o(1)}$.
\end{proof}

Finally, this establishes \thm{QvsUC}.

\QvsUC*

\section*{Acknowledgements}

I thank Shalev Ben-David for several helpful conversations and comments on a preliminary draft. 

\noindent This work is supported by ARO grant number W911NF-12-1-0486. 

\noindent This preprint is MIT-CTP \#4746.


\bibliographystyle{alpha}
\bibliography{partition}


\end{document}